\pdfoutput=1
\documentclass{article}
\usepackage{amsmath,amssymb,bbm,color,a4wide}

\newcommand{\tmem}[1]{{\em #1\/}}
\newcommand{\tmop}[1]{\ensuremath{\operatorname{#1}}}
\newcommand{\tmtextit}[1]{{\itshape{#1}}}
\newcommand{\um}{-}
\newcommand{\upl}{+}
\newenvironment{itemizedot}{\begin{itemize} }{\end{itemize}}
\newenvironment{proof}{\noindent\textbf{Proof\ }}{\hspace*{\fill}$\Box$\medskip}
\definecolor{grey}{rgb}{0.75,0.75,0.75}
\definecolor{orange}{rgb}{1.0,0.5,0.5}
\definecolor{brown}{rgb}{0.5,0.25,0.0}
\definecolor{pink}{rgb}{1.0,0.5,0.5}
\newtheorem{lemma}{Lemma}
\newtheorem{theorem}{Theorem}

\usepackage{authblk}
\title{One-side Energy Costs of the RBO Receiver\thanks{This work is supported
by ``Detectors and sensors for measuring factors hazardous to environment --
modeling and monitoring of threats'', a project financed by the European Union
via the European Regional Development Fund and the Polish state budget, within
the framework of the Operational Programme Innovative Economy 2007-2013. Ref.
No. POIG.01.03.01-02-002/08-00. }}
\author{Marcin Kik}
\author{Maciej G\c{e}bala}
\author{Miros{\l}aw Kuty{\l}owski}
\affil{Faculty of Fundamental Problems of Technology\\
Wroc{\l}aw University of Technology\\
ul. Wybrze\.ze Wyspia\'nskiego 27\\
PL-50-370 Wroc{\l}aw\\
Poland}
\begin{document}
\maketitle
\begin{abstract}
  Let $n = 2^k$ be the length of the broadcast cycle of the RBO broadcast
  scheduling protocol (see {\cite{DBLP:journals/corr/abs-1108-5095}},
  {\cite{DBLP:journals/corr/abs-1201-3318}}). Let $\tmop{lb}$ and $\tmop{ub}$
  be the variables of the RBO receiver as defined in
  {\cite{DBLP:journals/corr/abs-1201-3318}}. We show that the number of
  changes of $\tmop{lb}$ (the {\tmem{left-side energy}}) is not greater than
  $k + 1$. We also show that the number of changes of $\tmop{rb}$ (the
  {\tmem{right-side energy}}) is not greater than $k + 2$. Thus the
  {\tmem{extra energy}} \ (defined in
  {\cite{DBLP:journals/corr/abs-1201-3318}}) is bounded by $2 k + 3$. This
  updates the previous bound from {\cite{DBLP:journals/corr/abs-1201-3318}}
  which was $4 k + 2$.
\end{abstract}

\section{Introduction}

In this report we prove precise upper bound on the {\tmem{extra energy}} \
(defined in {\cite{DBLP:journals/corr/abs-1201-3318}}) used by the RBO
receiver. RBO broadcast scheduling protocol has been introduced in
{\cite{DBLP:journals/corr/abs-1108-5095}} and
{\cite{DBLP:journals/corr/abs-1201-3318}}. Following the description from
{\cite{DBLP:journals/corr/abs-1201-3318}}, the general operation of RBO could
be shortly outlined as follows:

A powerful {\tmem{broadcaster}} broadcasts repeatedly a {\tmem{broadcast
cycle}}. The broadcast cycle is a sequence of $n = 2^k$ uniform
{\tmem{messages}}. Transmission of each message occupies a single
{\tmem{time-slot}}. Each message contains a {\tmem{key}}, the {\tmem{index}}
of the key in the {\tmem{sorted}} sequence of all keys and possibly some other
information. The broadcast cycle is a sequence of the messages sorted by the
keys and permuted by $k$-bit reversal permutation.

On the other hand the receiver has a limited source of energy. It has to keep
its radio switched off most of the time to save energy. The receiver may start
at arbitrary time slot during the broadcast cycle and is interested in the
reception of the messages with the keys from some given interval $[\kappa',
\kappa'']$. It uses two variables: $\tmop{lb}$ and $\tmop{ub}$, initiated to
zero and $n - 1$, respectively, such that the indexes of the keys from \
$[\kappa', \kappa'']$ are contained in $[\tmop{lb}, \tmop{ub}]$. The receiver
switches its radio on whenever a message with the key with the index from
$[\tmop{lb}, \tmop{ub}]$ is transmitted. If the key is outside $[\kappa',
\kappa'']$, then the receiver updates either $\tmop{lb}$ or $\tmop{ub}$ and
the unit of energy used for its reception is accounted to so called
{\tmem{extra energy}}. After at most $n$ time slots, the interval $[\tmop{lb},
\tmop{ub}]$ contains only all the indexes of the keys from $[\kappa',
\kappa'']$. We show that the extra energy is bound by $2 k + 3$. This updates
the previous upper bound from {\cite{DBLP:journals/corr/abs-1201-3318}} which
was $4 k + 2$.

\section{Notation and preliminaries}

Let $\mathbbm{Z}$ denote the set of integers. Let $\mathbbm{R}$ denote the set
of real numbers. For simplicity and generality, we assume that the keys are
from $\mathbbm{R}$. By \ $[a, b]$ we denote the interval of real numbers $\{x
\in \mathbbm{R} \; | \; a \leq x \leq b\}$. By $[[a, b]]$ we denote $[a, b]
\cap \mathbbm{Z}$ (i.e. interval of integers between $a$ and $b$). For a
finite set $S$, we denote the number of its elements by $|S|$. Empty set is
denoted by $\emptyset$ and we have \ $\min \emptyset = \upl \infty$ and $\max
\emptyset = \um \infty$.

For $x \geq 0$, $x \in \mathbbm{Z}$, \ let $\tmop{bin}_i (x)$ denote the
sequence of $i$ digits that is binary representation of $x \tmop{mod} 2^i$
(e.g. $\tmop{bin}_4 (5) = (0101)$ -- sequence of zeroes and ones in
{\tmem{parenthesis}} denotes {\tmem{binary}} representation). For a sequence
$\alpha$ that is a binary representation, $(\alpha)_2$ denotes the number
represented by $\alpha$. Thus $(\tmop{bin}_i (x))_2 = x \tmop{mod} 2^i$. Let
$\tmop{bin} (x) = \tmop{bin}_{\lceil \lg_2 (x + 1) \rceil} (x)$. (Note that \
$\tmop{bin} (0) = \tmop{bin}_0 (0)$ is an empty sequence and, hence, for empty
sequence $\alpha$, we have $(\alpha)_2 = 0$.) For binary representation
$\alpha$,
\begin{itemize}
  \item $\alpha^d$ denotes the concatenation of $d$ copies of $\alpha$ (e.g.
  $(1)^4 = (1111)$),
  
  \item $\tmop{rev} \alpha$ denotes reversal of $\alpha$ (e.g. $\tmop{rev}
  (01) = (10)$), and
  
  \item $| \alpha |$ denotes the length \ of $\alpha$.
\end{itemize}
For two binary representations $\alpha$ and $\beta$, $\alpha \beta$ denotes
the concatenation of $\alpha$ and $\beta$ (e.g. $(01) (11) = (0111)$ ). \

For $x \in \mathbbm{Z}$, let $\tmop{rev}_k (x) = (\tmop{rev} \tmop{bin}_k
(x))_2$. Note that $\tmop{rev}_k (x)$ is the number with binary representation
that is a reversal of the $k$-bit binary representation of $x \tmop{mod} 2^k$.
Thus, $\tmop{rev}_k (x) \in [[0, 2^k - 1]]$. Note that, if $k > 0$, then
$\tmop{rev}_k (1) = 2^{k - 1}$. For a set $S \subseteq \mathbbm{Z}$,
$\tmop{rev}_k S$ denotes the {\tmem{image}} of $S$ under $\tmop{rev}_k$, i.e
$\tmop{rev}_k S =\{\tmop{rev}_k (x) | x \in S\}$.

Let $n$ denote the length of the broadcast cycle, $n = 2^k$, \ for integer $k
\geq 0$. Let $\kappa_{- 1}, \kappa_0, \ldots, \kappa_{n - 1}, \kappa_n$ be a
sequence defined as follows:
\begin{itemizedot}
  \item $\kappa_{- 1} = \um \infty$
  
  \item $\kappa_n = \upl \infty$
  
  \item $\kappa_0, \ldots, \kappa_{n - 1}$ is a sorted sequence of $n$ \
  finite real values of the keys transmitted in the broadcast cycle (i.e.
  $\kappa_i \leq \kappa_{i + 1}$, for \ \ \ $\um 1 \leq i \leq n - 1$).
\end{itemizedot}
Let $\tmop{KEYS} =\{\kappa_0, \ldots, \kappa_{n - 1} \}$. $\tmop{KEYS}$ is the
set of all values of the keys transmitted in the broadcast cycle.

\subsection{Outline of the RBO protocol}

Here, we shortly remind the version of the protocol described in
{\cite{DBLP:journals/corr/abs-1201-3318}}:

Broadcasting starts at the time-slot number zero. The broadcaster, at the
time-slot $t$, broadcasts the message with the key with the index
$\tmop{rev}_k (t)$. The receiver, requesting the keys from $[\kappa',
\kappa'']$, where $- \infty < \kappa' \leq \kappa'' < + \infty$, \ has two
variables $\tmop{lb}$ and $\tmop{ub}$ initialized to $0$ and \ $n - 1$,
respectively. The receiver may start at arbitrary time slot $s$, and executes
the following algorithm:
\begin{itemize}
  \item In time-slot $t$, if $\tmop{lb} \leq \tmop{rev}_k (t) \leq \tmop{ub}$,
  then the receiver receives the message with the key $\kappa =
  \kappa_{\tmop{rev}_k (t)}$ and
  \begin{itemize}
    \item if $\kappa < \kappa'$, then it sets $\tmop{lb}$ to $\tmop{rev}_k (t)
    + 1$,
    
    \item if $\kappa'' < \kappa$, then it sets $\tmop{ub}$ to $\tmop{rev}_k
    (t) - 1$,
    
    \item if $\kappa' \leq \kappa \leq \kappa''$, then it reports reception
    of the key $\kappa$ from $[\kappa', \kappa'']$
    
    \item if $\tmop{lb} > \tmop{ub}$, then it reports that there are no keys
    from $[\kappa', \kappa'']$ in the broadcast cycle.
  \end{itemize}
\end{itemize}

\subsection{Notation and preliminaries related to the analysis of the
protocol}

We fix $n$, $k$ and $s$ as follows: $n = 2^k$ is the length of the broadcast
cycle, and $s$ is the time-slot when the RBO receiver starts. W.l.o.g. we
assume that $s \in [[0, n - 1]]$.

Let $r'$ and $r''$ be defined as follows:
\begin{itemizedot}
  \item $r' = \min \{r \in [[0, n]] \; | \; \kappa' \leq \kappa_r \}$, and
  
  \item $r'' = \max \{r \in [[\um 1, n - 1]] \; | \; \kappa_r \leq \kappa''
  \}$.
\end{itemizedot}
For each $r \in [[r', r'']]$, we have $\kappa_r \in [\kappa', \kappa'']$. If
$[\kappa', \kappa''] \cap \tmop{KEYS} = \emptyset$ then, for some $r \in [[\um
1, n - 1]]$, we have $\kappa_r < \kappa'$ and $\kappa'' < \kappa_{r + 1}$ and
$r' = r + 1$ and $r'' = r$. \ If $[\kappa', \kappa''] \cap \tmop{KEYS} \not=
\emptyset$ then, since $\um \infty < \kappa' \leq \kappa'' < \upl \infty$, we
have \ $0 \leq r' \leq r'' \leq n - 1$. Note that $r'$ and $r''$ are the final
values of the variables $\tmop{lb}$ and $\tmop{ub}$, respectively.

We also define of the values: \ $t_i$, $l_i$, and the sets: $Y_i$, $X_i$,
$Y_{i, j}$, and $X_{i, j}$ as follows: $t_0 = s$ and, for $i \geq 0$,
\begin{itemize}
  \item  $l_i = \max \{l \leq k \; | \; t_i \tmop{mod} 2^l = 0\}$ (i.e. $l_i$
  is the length of the longest suffix of zeroes in $\tmop{bin}_k (t_i)$),
  
  \item  $t_{i + 1} = t_i + 2^{l_i}$,
  
  \item $Y_i = [[t_i, t_i + 2^{l_i} - 1]]$,
  
  \item  $X_i = \tmop{rev}_k Y_i$,
\end{itemize}
and, for $i \geq 0$, for $j \in [[0, l_i]]$,
\begin{itemize}
  \item $Y_{i, j} = [[t_i + \lfloor 2^{j - 1} \rfloor, t_i + 2^j - 1]]$ (in
  other words: $Y_{i, 0} =\{t_i \}$ and, for $j \in [[1, l_i]]$, $Y_{i, j} =
  [[t_i + 2^{j - 1}, t_i + 2^j - 1]]$),
  
  \item $X_{i, j} = \tmop{rev}_k Y_{i, j}$.
\end{itemize}
We also define index $\tmop{last}$ as follows: $\tmop{last} = \min \{i \geq 0
| l_i = k\}$. \

$Y_i$ and $Y_{i, j}$ are intervals of the time-slots after $s - 1$. The
infinite sequence of the time-slots $t_s, t_{s + 1}, \ldots$ \ is partitioned
into the intervals $Y_0, Y_1, \ldots$, \ and each $Y_i$ is partitioned into
the intervals $Y_{i, 0}, \ldots, Y_{i, l_i}$. $X_i$ and $X_{i, j}$ are the
sets of indexes of the keys transmitted during the corresponding intervals of
time-slots. Note that, if $s \not= 0$, then for $i \in [[0, \tmop{last} -
1]]$, $Y_i \subseteq [[0, n - 1]]$, and $Y_{\tmop{last}} = [[n, 2 n - 1]]$,
and if $s = 0$, then $\tmop{last} = 0$ and $Y_{\tmop{last}} = [[0, n - 1]]$.
We have $s + n - 1 \in Y_{\tmop{last}}$ and, until the time slot $s + n - 1$,
\ all the elements of the broadcast cycle are transmitted to the receiver.
Note also that $[[0, n - 1]]$ contains a {\tmem{disjoint}} union of the sets
$X_0, \ldots, X_{\tmop{last} - 1}$, and $X_{\tmop{last}} = [[0, n - 1]]$, and
each $X_i$ is a {\tmem{disjoint}} union of the sets $X_{i, 0}, \ldots, X_{i,
l_i}$.

Since all $k$-bit binary representations of $X_i$ have the same suffix of
length $k - l_i$, we may define $\beta_i$ as follows: For $i \in [[0,
\tmop{last}]]$, let the sequence of bits $\beta_i$ be such that $| \beta_i | =
k - l_i$ and $\min X_i = \tmop{rev}_k (t_i) = ((0)^{l_i} \beta_i)_2$. Note
that $X_i =\{(\tmop{bin} (x) \beta_i)_2 | x \in [[0, 2^{l_i} - 1]]\}$.

For $i \in [[0, \tmop{last} - 2]]$, the sequence $\beta_i$ contains prefix
$(1)^{l_{i + 1} - l_i} (0)$, so let $\alpha_i$ be such that $\beta_i =
(1)^{l_{i + 1} - l_i} (0) \alpha_i$.

Note that $\max X_i = ((1)^{l_i} \beta_i)_2$ and, if $i + 1 < \tmop{last}$,
then $\beta_{i + 1} = (1) \alpha_i$. Note that $\beta_{\tmop{last} - 1} =
(1)^{l_{\tmop{last}} - l_{\tmop{last} - 1}}$ and $| \beta_{\tmop{last}} | =
0$. Anyway, for $i \in [[0, \tmop{last} - 1]]$, we have $\beta_i = (1)^{l_{i +
1} - l_i} \gamma$, where $(\gamma)_2 \leq (\beta_{i + 1})_2$. Thus we have
proven the following lemma:

\begin{lemma}
  \label{suffix-lemma} If $i \in [[0, \tmop{last} - 1]]$ \ and $\beta_i = \rho
  \gamma$, where $| \gamma | = | \beta_{i + 1} |$, then $(\gamma)_2 \leq
  (\beta_{i + 1})_2$.
\end{lemma}

Consider the binary representations of the elements of $X_{i, j}$: Note that,
for $i \in [[0, \tmop{last}]]$, for $j \in [[0, l_i]]$, we have $x \in X_{i,
j}$ if and only if either $j = 0$ and $\tmop{bin}_k (x) = (0)^{l_i} \beta_i$ \
or \ $j \in [[1, l_i]]$, and \ $\tmop{bin}_k (x) = \rho (1) (0)^{l_i - j}
\beta_i$, where $(\rho)_2 \in [[0, 2^{j - 1} - 1]]$.

We also define the following infinite extensions of the sets of indexes: For
$i \in [[0, \tmop{last}]]$, let $\mathbbm{X}_i =\{x \in \mathbbm{Z} \; | \; x
\tmop{mod} 2^{k - l_i} = (\beta_i)_2 \}$. Note that $\mathbbm{X}_i \cap [[0, n
- 1]] = X_i$. For $i \in [[0, \tmop{last}]]$, for $j \in [[0, l_i]]$, let
$\mathbbm{X}_{i, j} =\{x \in \mathbbm{Z} \; | \; x \tmop{mod} 2^{k - j} =
((0)^{l_i - j} \beta_i)_2 \}$. Note that $\mathbbm{X}_{i, j} \cap [[0, n - 1]]
= \bigcup_{j' \in [[0, j]]} X_{i, j'}$.

By {\tmem{left-side energy}} (respectively, {\tmem{right-side energy}}) we
mean the number of changes of the variable $\tmop{lb}$ (respectively,
$\tmop{ub}$). Let $\tmop{lb}_t$ (respectively, $\tmop{ub}_t$) be the value of
the variable $\tmop{lb}$ \ (respectively, $\tmop{ub}$) just before the time
slot $t$. By the protocol, $\tmop{lb}_t = \max \left( \{0\} \cup
\{\tmop{rev}_k (t') + 1 \leq r' \; | \; t' \in [[s, t - 1]]\} \right)$ and
$\tmop{ub}_t = \min \left( \{n - 1\} \cup \{\tmop{rev}_k (t') - 1 \geq r'' \;
| \; t' \in [[s, t - 1]]\} \right)$.

For $t \geq 0$, let $L_t =\{\tmop{lb}_{t + 1} - 1\} \setminus \{\tmop{lb}_t -
1\}$ and let $U_t =\{\tmop{ub}_{t + 1} + 1\} \setminus \{\tmop{ub}_t + 1\}$.
If the value of $\tmop{lb}$ (respectively, of $\tmop{ub}$) has changed in the
time slot $t$, then $L_t$ (respectively, $U_t$) contains the index
$\tmop{rev}_k (t)$, otherwise it is empty. Hence, no index \ can appear in
more than one of these sets. \ The left-side energy is equal to $| \bigcup_{t
\geq s} L_t |$. The right-side energy is equal to $| \bigcup_{t \geq s} U_t
|$. Note that, for $t \not\in [[s, s + n - 1]]$, we have $L_t = U_t =
\emptyset$.

We consider only the case $k \geq 2$. (Otherwise, the sum of the left and
right energy is bound by 2.)

\section{Left-side energy}

Let $t' = \min \{t \in [[s, s + n - 1]] \; | \tmop{lb}_{t + 1} > 0\}$. Note
that $t'$ is the first time slot, when $\tmop{lb}$ is updated. If $r' > 0$,
then $t' < \infty$.

Let $m'_{i, j} = \max \{\tmop{lb}_{t + 1} - 1 \; | \; t \in Y_{i, j} \}$.
Note that $m'_{i, j}$ is either $\um 1$ or the maximal index that updated the
value of $\tmop{lb}$ before the time slot \ $\max Y_{i, j} + 1$, and $m'_{i,
j} = \tmop{lb}_{\max Y_{i, j} + 1} - 1$. Let $m'_i = m'_{i, l_i}$.

{\color{black} \begin{lemma}
  \label{m-lemma1}$m'_{i, j} = \max \left( \{\um 1\} \cup \{x \in \bigcup_{i'
  \in [[0, i - 1]]} X_{i'} \cup \bigcup_{j' \in [[0, j]]} X_{i, j'}  \; | \; x
  < r' \} \right)$. If $t' = s$, then $m'_{i, j} \geq 0$.
\end{lemma}

\begin{proof}
  $m'_{i, j} = \tmop{lb}_{\max Y_{i, j} + 1} - 1 = \max \left( \{\um 1\} \cup
  \{\tmop{rev}_k (t') < r'  \; | \; t' \in [[s, \max Y_{i, j}]]\} \right)$ and
  $\{\tmop{rev}_k (t') \; | \; t' \in [[s, \max Y_{i, j}]]\}= \bigcup_{i' \in
  [[0, i - 1]]} X_{i'} \cup \bigcup_{j' \in [[0, j]]} X_{i, j'}$.
  
  If $t' = s$, then $\tmop{lb}_{s + 1} - 1 \geq 0$. 
\end{proof}}

Let $p'_{i, j} = \max \{x \in \mathbbm{X}_{i, j}  \; | \; x < r' \}$. Note
that either $p'_{i, j} < 0$ or $p'_{i, j} \in \bigcup_{0 \leq j' \leq j} X_{i,
j'}$. {\color{black}  \ By Lemma~\ref{m-lemma1}, $m'_{i, j} \geq p'_{i, j}$.}
Let $p'_i = \max \{x \in \mathbbm{X}_i  \; | \; x < r' \}$. By
Lemma~\ref{m-lemma1}, $m'_i \geq p'_i$. Let $x'_i = \lfloor p'_i / 2^{k - l_i}
\rfloor$. Note that, by definition of $\mathbbm{X}_i$, we have $x'_i \geq -
1$.

\begin{lemma}
  \label{sublevel-lemma1}If $t' = s$, then, for $i \in [[0, \tmop{last}]]$,
  for $j \in [[0, l_i]]$, $\bigcup_{t \in Y_{i, j}} L_t \subseteq \{p'_{i, j}
  \} \cap X_{i, j}$.
\end{lemma}

\begin{proof}
  The case $j = 0$: We have $\bigcup_{t \in Y_{i, 0}} L_t \subseteq X_{i, 0}$.
  If $r' \leq \min X_{i, 0}$ then \ $\bigcup_{t \in Y_{i, 0}} L_t =
  \emptyset$. Otherwise, since $|X_{i, 0} | = 1$, we have $\min X_{i, 0} =
  \max X_{i, 0} < r'$ and, since $r' \leq n$ and $ X_{i, 0} =\mathbbm{X}_{i,
  0} \cap [[0, n - 1]]$ , we have $X_{i, 0} =\{p'_{i, 0} \}$.
  
  Induction step: For $j \in [[0, l_i - 1]]$, since $p'_{i, j} + 2^{k - j}
  \geq r'$, we have $X_{i, j + 1} \cap [[p'_{i, j} + 1, r' - 1]] \subseteq
  \{p'_{i, j} + 2^{k - j - 1} \}$. If $p'_{i, j} + 2^{k - j - 1} < r'$, then \
  $p'_{i, j + 1} = p'_{i, j} + 2^{k - j - 1}$, otherwise $X_{i, j + 1} \cap
  [[p'_{i, j} + 1, r' - 1]] = \emptyset$. Since $m'_{i, j} \geq p'_{i, j}$ and
  $\bigcup_{t \in Y_{i, j + 1}} L_t \subseteq X_{i, j + 1} \cap [[m'_{i, j} +
  1, r' - 1]]$, we have $\bigcup_{t \in Y_{i, j + 1}} L_t \subseteq \{p'_{i, j
  + 1} \} \cap X_{i, j + 1}$.
\end{proof}

\begin{lemma}
  \label{starting-lemma1}If $t' = s$, then \ $| \bigcup_{t \in Y_0} L_t | \leq
  l_0 + 1$.
\end{lemma}

\begin{proof}
  The lemma follows from \ Lemma~\ref{sublevel-lemma1}, since $Y_0 =
  \bigcup_{j \in [[0, l_0]]} Y_{0, j}$ and, for $j \in [[0, l_0]]$, $|
  \bigcup_{t \in Y_{0, j}} L_t | \leq 1$.
\end{proof}

\begin{lemma}
  \label{middle-lemma1}If $t' = s$, then, for $i \in [[0, \tmop{last} - 2]]$,
  $| \bigcup_{t \in Y_{i + 1}} L_t | \leq l_{i + 1} - l_i$.
\end{lemma}

\begin{proof}
  Note that $l_{i + 1} - l_i \geq 1$. Recall that $m'_i \geq p'_i = 2^{k -
  l_i} \cdot x'_i + (\beta_i)_2$, $\beta_i = (1)^{l_{i + 1} - l_i} (0)
  \alpha_i$ and $\beta_{i + 1} = (1) \alpha_i$.
  
  Consider the case: $r' \leq (\tmop{bin} (x'_i + 1) (0)^{l_{i + 1} - l_i}
  \beta_{i + 1})_2$. \ Then
  \begin{eqnarray*}
    \bigcup_{t \in Y_{i + 1}} L_t & \subseteq & X_{i + 1} \cap [[m'_i + 1, r'
    - 1]]\\
    & \subseteq & \mathbbm{X}_{i + 1} \cap [[p'_i + 1, r' - 1]]\\
    & \subseteq & \{2^{k - l_i} \cdot x_i' + ((1)^{l_{i + 1} - l_i} \beta_{i
    + 1})_2 \},
  \end{eqnarray*}
  since $p'_i = 2^{k - l_i} \cdot x'_i + ((1)^{l_{i + 1} - l_i} (0)
  \alpha_i)_2 > 2^{k - l_i} \cdot x'_i + ((1)^{l_{i + 1} - l_i - 1} (0)
  \beta_{i + 1})_2$ and $r' \leq (\tmop{bin} (x'_i + 1) (0)^{l_{i + 1} - l_i}
  \beta_{i + 1})_2$.
  
  Consider the case: $(\tmop{bin} (x'_i + 1) (0)^{l_{i + 1} - l_i} \beta_{i +
  1})_2 < r'$. Since
  \begin{itemize}
    \item  $2^{k - l_i} \cdot x'_i + ((0)^{l_{i + 1} - l_i} \beta_{i + 1})_2 <
    p'_i$, and
    
    \item $(\tmop{bin} (x'_i + 2) (0)^{l_{i + 1} - l_i} \beta_{i + 1})_2 >
    (\tmop{bin} (x'_i + 1) (1)^{l_{i + 1} - l_i} (0) \alpha_i)_2 = (\tmop{bin}
    (x'_i + 1) \beta_i)_2 \geq r'$, and
    
    \item each $x \in \bigcup_{l \in [[0, l_i]]} X_{i + 1, l}$ has
    $\tmop{bin}_k (x) = \xi (0)^{l_{i + 1} - l_i} \beta_{i + 1}$, for some
    $\xi$,
  \end{itemize}
  we have
  \[ \bigcup_{l \in [[0, l_i]]} X_{i + 1, l} \cap [[p'_i + 1, r' - 1]]
     =\{(\tmop{bin} (x'_i + 1) (0)^{l_{i + 1} - l_i} \beta_{i + 1})_2 \}. \]
  Hence, $m'_{i + 1, l_i} \geq (\tmop{bin}_{l_i} (x'_i + 1) (0)^{l_{i + 1} -
  l_i} \beta_{i + 1})_2$, and, since $\bigcup_{l \in [[0, l_i]]} \bigcup_{t
  \in Y_{i + 1, l}} L_t \subseteq \bigcup_{l \in [[0, l_i]]} X_{i + 1, l} \cap
  [[p'_i + 1, r' - 1]]$, we also have $| \bigcup_{l \in [[0, l_i]]} \bigcup_{t
  \in Y_{i + 1, l}} L_t | \leq 1$.
  
  By Lemma~\ref{sublevel-lemma1}, for each $l \in [[l_i + 1, l_{i + 1}]]$, we
  have $| \bigcup_{t \in Y_{i + 1, l}} L_t | \leq 1$.
  
  Moreover, if $| \bigcup_{l \in [[0, l_{i + 1} - 1]]} ( \bigcup_{t \in Y_{i
  + 1, l}} L_t) | = l_{i + 1} - l_i$, then, $m'_{i + 1, l_i} \in \bigcup_{l
  \in [[0, l_i]]} \bigcup_{t \in Y_{i + 1, l}} L_t$, and, for each $l \in
  [[l_i + 1, l_{i + 1} - 1]]$, $m'_{i + 1, l} \in X_{i + 1, l} \cap [[m'_{i +
  1, l - 1} + 1, r' - 1]]$. The only such case is $m'_{i + 1, l_i} =
  (\tmop{bin} (x'_i + 1) (0)^{l_{i + 1} - l_i} \beta_{i + 1})_2$, and, for
  each $l \in [[l_i + 1, l_{i + 1} - 1]]$, $m'_{i + 1, l_{}} = (\tmop{bin}
  (x'_i + 1) (1)^{l - l_i} (0)^{l_{i + 1} - l} \beta_{i + 1})_2$. In this
  case, since $(\tmop{bin} (x'_i + 1) (1)^{l_{i + 1} - l_i} \beta_{i + 1})_2 >
  (\tmop{bin} (x'_i + 1) (1)^{l_{i + 1} - l_i} (0) \alpha_i)_2 = (\tmop{bin}
  (x'_i + 1) \beta_i)_2 \geq r'$, we have $X_{i + 1, l_{i + 1}} \cap [[m'_{i +
  1, l_{i + 1} - 1} + 1, r' - 1]] = \emptyset$ and, hence, $| \bigcup_{t \in
  Y_{i + 1, l_{i + 1}}} L_t | = 0$.
\end{proof}

\begin{lemma}
  \label{last-lemma1}If $\tmop{last} > 0$ and $t' = s$, then $| \bigcup_{t \in
  Y_{\tmop{last}}} L_t | \leq l_{\tmop{last}} - l_{\tmop{last} - 1}$.
\end{lemma}

\begin{proof}
  If $r' \leq (\tmop{bin} (x'_{\tmop{last} - 1} + 1) (0)^{l_{\tmop{last}} -
  l_{\tmop{last} - 1}})_2$, then, since $m'_{\tmop{last} - 1} \geq
  p'_{\tmop{last} - 1} = 2^{l_{\tmop{last}} - l_{\tmop{last} - 1}} \cdot
  x'_{\tmop{last} - 1} + ((1)^{l_{\tmop{last}} - l_{\tmop{last} - 1}})_2$, we
  have $X_{\tmop{last}} \cap [[m'_{\tmop{last} - 1} + 1, r' - 1]] =
  \emptyset$, and, hence $| \bigcup_{t \in Y_{\tmop{last}}} L_t | = 0$. \
  
  Otherwise, we have the case: $(\tmop{bin} (x'_{\tmop{last} - 1} + 1)
  (0)^{l_{\tmop{last}} - l_{\tmop{last} - 1}})_2 < r'$. In this case
  \[ \bigcup_{l \in [[0, l_{\tmop{last} - 1}]]} X_{\tmop{last}, l} \cap
     [[p'_{\tmop{last} - 1} + 1, r' - 1]] =\{(\tmop{bin} (x_{\tmop{last} - 1}'
     + 1) (0)^{l_{\tmop{last}} - l_{\tmop{last} - 1}})_2 \}, \]
  since $2^{k - l_{\tmop{last} - 1}} \cdot x'_{\tmop{last} - 1} + ((0)^{l_{i +
  1} - l_i})_2 < p'_{\tmop{last} - 1}$ and \ $(\tmop{bin} (x_{\tmop{last} -
  1}' + 2) (0)^{l_{\tmop{last}} - l_{\tmop{last} - 1}})_2 > (\tmop{bin}
  (x_{\tmop{last} - 1}' + 1) (1)^{l_{\tmop{last}} - l_{\tmop{last} - 1}})_2
  \geq r'$. Hence, $| \bigcup_{l \in [[0, l_{\tmop{last} - 1}]]} \bigcup_{t
  \in Y_{\tmop{last}, l}} L_t | \leq 1$ and $m'_{\tmop{last}, l_{\tmop{last} -
  1}} \geq (\tmop{bin} (x' + 1) (0)^{l_{\tmop{last}} - l_{\tmop{last} -
  1}})_2$. By Lemma~\ref{sublevel-lemma1}, for $l \in [[l_{\tmop{last} - 1} +
  1, l_{\tmop{last}} - 1]]$, we have $| \bigcup_{t \in Y_{i + 1, l}} L_t |
  \leq 1$. \ Moreover, if $| \bigcup_{l \in [[0, l_{\tmop{last}} - 1]]} (
  \bigcup_{t \in Y_{\tmop{last}, l}} L_t) | = l_{\tmop{last}} - l_{\tmop{last}
  - 1}$, then $m'_{\tmop{last}, l_{\tmop{last}} - 1} = (\tmop{bin}
  (x'_{\tmop{last} - 1} + 1) (1)^{l_{\tmop{last}} - 1 - l_{\tmop{last} - 1}}
  (0))_2$ and, since $(\tmop{bin} (x'_{\tmop{last} - 1} + 1)
  (1)^{l_{\tmop{last}} - l_{\tmop{last} - 1}})_2 \geq r'$, we have
  \[ X_{\tmop{last}, l_{\tmop{last}}} \cap [[m'_{\tmop{last}, l_{\tmop{last}}
     - 1} + 1, r' - 1]] = \emptyset . \]
  
\end{proof}

\begin{lemma}
  \label{left-side-lemma} $| \bigcup_{t \geq s} L_t | \leq k + 1$.
\end{lemma}

\begin{proof}
  If $t'  \not\in [[s, s + n - 1]]$ then \ $t' = \infty$ and $r' = 0$ and $|
  \bigcup_{t \geq s} L_t | = 0 < k + 1$.
  
  Otherwise $t' \in [[s, s + n - 1]]$. \ Note that w.l.o.g. we may assume
  that $t' = s$, since $| \bigcup_{t \geq s} L_t | = | \bigcup_{t \geq t'} L_t
  |$. So, let us assume that $t' = s$. We have $| \bigcup_{t \geq s} L_t | =
  \sum_{i = 0}^{\tmop{last}} | \bigcup_{t \in Y_i} L_t |$. By
  Lemmas~\ref{starting-lemma1}, \ref{middle-lemma1}, and \ref{last-lemma1},
  $\sum_{i = 0}^{\tmop{last}} | \bigcup_{t \in Y_i} L_t | \leq l_0 + 1 +
  \sum_{i = 0}^{\tmop{last} - 1} (l_{i + 1} - l_i) = l_{\tmop{last}} + 1 \leq
  k + 1$.
\end{proof}

\section{Right-side energy}

Let $t'' = \min \{t \in [[s, s + n - 1]] \; | \; \tmop{ub}_{t + 1} + 1 \leq n
- 1\}$ (i.e. the first time slot when $\tmop{ub}$ is updated). If $r'' < n -
1$, then $t'' < \infty$.

Let $m''_{i, j} = \min \{\tmop{ub}_{t + 1} + 1 \; | \; t \in Y_{i, j} \}$.
Note that $m''_{i, j} = \tmop{ub}_{\max Y_{i, j} + 1} + 1$. Let $m''_i =
m''_{i, l_i}$.

\begin{lemma}
  \label{m-lemma2}$m''_{i, j} = \min \left( \{n\} \cup \{x \in \bigcup_{i' \in
  [[0, i - 1]]} X_{i'} \cup \bigcup_{j' \in [[0, j]]} X_{i, j'} \; | \; r'' <
  x\} \right)$. If $t'' = s$, then $m''_{i, j} \leq n - 1$.
\end{lemma}

\begin{proof}
  The proof is similar to the proof of Lemma~\ref{m-lemma1}.
\end{proof}

Let $p''_{i, j} = \min \{x \in \mathbbm{X}_{i, j} \; | \; r'' < x\}$. By
Lemma~\ref{m-lemma2}, $m''_{i, j} \leq p''_{i, j}$. Let $p''_i = \min \{x \in
\mathbbm{X}_i  \; | \; r'' < x\}$. By Lemma~\ref{m-lemma2}, $m''_i \leq
p''_i$. Let $x''_i = \lfloor p''_i / 2^{k - l_i} \rfloor$. Since $r'' \geq -
1$, we have $p''_i \geq 0$ and $x''_i \geq 0$.

\begin{lemma}
  \label{sublevel-lemma2}If $t'' = s$, then, for $i \in [[0, \tmop{last}]]$,
  for $j \in [[0, l_i]]$, $\bigcup_{t \in Y_{i, j}} U_t \subseteq \{p''_{i, j}
  \} \cap X_{i, j}$.
\end{lemma}

\begin{proof}
  The proof is similar to the proof of \ Lemma~\ref{sublevel-lemma1}.
\end{proof}

\begin{lemma}
  \label{starting-lemma2}If $t'' = s$ then $| \bigcup_{t \in Y_0} U_t | = 1$.
\end{lemma}

\begin{proof}
  Note that $t'' = s = t_0 = \min Y_0$ and $\tmop{rev}_k (\min Y_0) = \min
  X_0$.
\end{proof}

\begin{lemma}
  \label{middle-lemma2}If $t'' = s$, then, for $i \in [[0, \tmop{last} - 2]]$,
  \ $| \bigcup_{t \in Y_{i + 1}} U_t | \leq \max \{l_{i + 1} - l_i, 2\}$.
\end{lemma}

\begin{proof}
  Recall that $l_{i + 1} - l_i \geq 1$ and that $m''_i \leq p''_i =
  (\tmop{bin} (x''_i) \beta_i)_2$. Thus, we have $(x''_i - 1) \cdot 2^{k -
  l_i} + (\beta_i)_2 \leq r''$. Recall that $\beta_i = (1)^{l_{i + 1} - l_i}
  (0) \alpha_i$ and $\beta_{i + 1} = (1) \alpha_i$.
  
  Consider the case: $r'' < (\tmop{bin} (x''_i) (0)^{l_{i + 1} - l_i} \beta_{i
  + 1})_2$. We show that in this case \ $\bigcup_{t \in Y_{i + 1}} U_t
  \subseteq \{a, b\}$, where $a = (\tmop{bin} (x''_i) (0)^{l_{i + 1} - l_i}
  \beta_{i + 1})_2$ and $b = (x''_i - 1) \cdot 2^{k - l_i} + ((1)^{l_{i + 1} -
  l_i} \beta_{i + 1})_2 $.
  
  First, note that \ $\bigcup_{l \in [[0, l_i]]} \bigcup_{t \in Y_{i + 1, l}}
  U_t \subseteq \bigcup_{l \in [[0, l_i]]} X_{i + 1, l} \cap [[r'' + 1, m''_i
  - 1]]$ and $m''_i \leq p''_i$. We have $\bigcup_{l \in [[0, l_i]]} X_{i + 1,
  l} \subseteq \mathbbm{X}_{i + 1, l_i}$, and the only $x \in \mathbbm{X}_{i +
  1, l_i}$, such that $(x''_i - 1) \cdot 2^{k - l_i} + (\beta_i)_2 < x <
  (\tmop{bin} (x''_i) \beta_i)_2$ is $x = (\tmop{bin} (x''_i) (0)^{l_{i + 1} -
  l_i} \beta_{i + 1})_2 = a$. (Note that this also implies $m''_{i + 1, l_i}
  \leq a$.)
  
  Second, note that $\bigcup_{l \in [[l_i + 1, l_{i + 1}]]} \bigcup_{t \in
  Y_{i + 1, l}} U_t \subseteq B$, where $B = \bigcup_{l \in [[l_i + 1, l_{i +
  1}]]} X_{i + 1, l} \cap [[r'' + 1, m''_{i + 1, l_i} - 1]]$ and, since
  $(x''_i - 1) \cdot 2^{k - l_i} + ((1)^{l_{i + 1} - l_i} (0) \alpha_i)_2 \leq
  r''$ and \ $m''_{i + 1, l_i} \leq (\tmop{bin} (x''_i) (0)^{l_{i + 1} - l_i}
  \beta_{i + 1})_2$, we have \ $B \subseteq \{(x''_i - 1) \cdot 2^{k - l_i} +
  ((1)^{l_{i + 1} - l_i} \beta_{i + 1})_2 \}=\{b\}$.
  
  Consider the case: $(\tmop{bin} (x''_i) (0)^{l_{i + 1} - l_i} \beta_{i +
  1})_2 \leq r''$. We have $\bigcup_{t \in Y_{i + 1}} U_t \subseteq C$, where
  $C = X_{i + 1} \cap [[r'' + 1, m''_i - 1]]$. Since $(\tmop{bin} (x''_i)
  (0)^{l_{i + 1} - l_i} \beta_{i + 1})_2 \leq r''$ and $m''_i \leq p''_i <
  (\tmop{bin} (x''_i) (1)^{l_{i + 1} - l_i} \beta_{i + 1})_2$, we have $C
  =\{(\tmop{bin} (x''_i) \gamma \beta_{i + 1})_2  \; | \; | \gamma | = l_{i +
  1} - l_i \wedge (\gamma)_2 \in [[1, 2^{l_{i + 1} - l_i} - 2]]\}$. Since, for
  $\gamma = (0)^{l_{i + 1} - l_i}$, $(\tmop{bin} (x''_i) \gamma \beta_{i +
  1})_2 \not\in C$, we have \ $C \subseteq \bigcup_{l \in [[l_i + 1, l_{i +
  1}]]} X_{i + 1, l}$. Thus, $\bigcup_{l \in [[0, l_i]]} \bigcup_{t \in Y_{i +
  1, l}} U_t = \emptyset$. By Lemma~\ref{sublevel-lemma2}, \ $| \bigcup_{l \in
  [[l_i + 1, l_{i + 1}]]} \bigcup_{t \in Y_{i + 1}, l} U_t | \leq l_{i + 1} -
  l_i$. Thus, $| \bigcup_{t \in Y_{i + 1}} U_t | \leq l_{i + 1} - l_i$ in this
  case.
\end{proof}

\begin{lemma}
  \label{last-lemma2}If \ $t'' = s$ \ and \ $0 < \tmop{last}$, then $|
  \bigcup_{t \in Y_{\tmop{last}}} U_t | \leq l_{\tmop{last}} - l_{\tmop{last}
  - 1}$.
\end{lemma}

\begin{proof}
  Recall that $m''_{\tmop{last} - 1} \leq p''_{\tmop{last} - 1} = (\tmop{bin}
  (x_{\tmop{last} - 1}'') (1)^{l_{\tmop{last}} - l_{\tmop{last} - 1}})_2$. We
  have $r'' \geq (x_{\tmop{last} - 1}'' - 1) \cdot 2^{k - l_{\tmop{last} - 1}}
  + ((1)^{l_{\tmop{last}} - l_{\tmop{last} - 1}})_2$.
  
  Consider the case $r'' < (\tmop{bin} (x''_{\tmop{last} - 1})
  (0)^{l_{\tmop{last}} - l_{\tmop{last} - 1}})_2$. Then $r'' = (x_{\tmop{last}
  - 1}'' - 1) \cdot 2^{k - l_{\tmop{last} - 1}} + ((1)^{l_{\tmop{last}} -
  l_{\tmop{last} - 1}})_2$ and, since $\tmop{rev}_k (\min \{t \in
  Y_{\tmop{last}} \; | \; r'' < \tmop{rev}_k (t) < p''_{\tmop{last} - 1} \}) =
  (\tmop{bin} (x''_{\tmop{last} - 1}) (0)^{l_{\tmop{last}} - l_{\tmop{last} -
  1}})_2 = r'' + 1$, we have $\bigcup_{t \in Y_{\tmop{last}}} U_t \subseteq
  \{r'' + 1\}$. (In other words: the first time slot $t$ in $Y_{\tmop{last}}$
  such that $r'' < \tmop{rev}_k (t) < p''_{\tmop{last} - 1}$ is such that
  $\tmop{rev}_k (t) = r'' + 1$.)
  
  Otherwise, we have the case $(\tmop{bin} (x''_{\tmop{last} - 1})
  (0)^{l_{\tmop{last}} - l_{\tmop{last} - 1}})_2 \leq r''$. We also have
  $m''_{\tmop{last} - 1} \leq (\tmop{bin} (x''_{\tmop{last} - 1})
  (1)^{l_{\tmop{last}} - l_{\tmop{last} - 1}})_2$. Thus,
  \[ \bigcup_{t \in Y_{\tmop{last}}} U_t \subseteq \{(\tmop{bin}
     (x''_{\tmop{last} - 1}) \gamma)_2 \; | \; | \gamma | = l_{\tmop{last}} -
     l_{\tmop{last} - 1} \wedge (\gamma)_2 > 0\}. \]
  Hence, $\bigcup_{t \in Y_{\tmop{last}}} U_t \subseteq \bigcup_{l \in
  [[l_{\tmop{last} - 1} + 1, l_{\tmop{last}}]]} X_{\tmop{last}, l}$. Thus,
  $\bigcup_{l \in [[0, l_{\tmop{last} - 1}]]} \bigcup_{t \in Y_{\tmop{last},
  l}} U_t = \emptyset$. By Lemma~\ref{sublevel-lemma2}, \ $| \bigcup_{l \in
  [[l_{\tmop{last} - 1} + 1, l_{\tmop{last}}]]} \bigcup_{t \in Y_{i + 1}, l}
  U_t | \leq l_{\tmop{last}} - l_{\tmop{last} - 1}$. \
  
  Thus, in any case, $| \bigcup_{t \in Y_{\tmop{last}}} U_t | \leq \max \{1,
  l_{\tmop{last}} - l_{\tmop{last} - 1} \}= l_{\tmop{last}} - l_{\tmop{last} -
  1}$.
\end{proof}

\begin{lemma}
  \label{steps-lemma}If $t'' = s$ and, for some $i \in [[0, \tmop{last} -
  1]]$, $x''_{i + 1} \geq (\tmop{bin} (x''_i) (0)^{l_{i + 1} - l_i - 1}
  (1))_2$ and $| \bigcup_{t \in Y_{i + 1}} U_t | \geq l_{i + 1} - l_i$, then
  $| \bigcup_{t \in Y_{i + 1}} U_t | = l_{i + 1} - l_i$ and $x''_{i + 1} =
  (\tmop{bin} (x''_i) (0)^{l_{i + 1} - l_i - 1} (1))_2$ and $m''_{i + 1} =
  p''_{i + 1}$.
\end{lemma}

\begin{proof}
  We have $l_{i + 1} - l_i \geq 1$, and $(\tmop{bin} (x''_i) (0)^{l_{i + 1} -
  l_i} \beta_{i + 1})_2 \leq (\tmop{bin} (x''_{i + 1} - 1) \beta_{i + 1})_2
  \leq r'' < p''_i = (\tmop{bin} (x''_i) (1)^{l_{i + 1} - l_i} \gamma)_2$,
  where $| \gamma | = | \beta_{i + 1} |$ and, by Lemma~\ref{suffix-lemma},
  $(\gamma)_2 \leq (\beta_{i + 1})_2$.
  
  We have
  \[ \bigcup_{l \in [[0, l_i]]} \bigcup_{t \in Y_{i + 1, l}} U_t \subseteq
     \bigcup_{l \in [[0, l_i]]} X_{i + 1, l} \cap [[r'' + 1, m''_i - 1]] =
     \emptyset, \]
  since $(\tmop{bin} (x''_i) (0)^{l_{i + 1} - l_i} \beta_{i + 1})_2 \leq r''$
  and $m''_i \leq p''_i < (\tmop{bin} (x''_i + 1) (0)^{l_{i + 1} - l_i}
  \beta_{i + 1})_2$. Thus, we have $m''_{i + 1, l_i} = m''_i$.
  
  Consider the case $| \bigcup_{l \in [[l_i + 1, l_{i + 1}]]} \bigcup_{t \in
  Y_{i + 1, l}} U_t | \geq l_{i + 1} - l_i$: By Lemma~\ref{sublevel-lemma2},
  for each $l \in [[l_i + 1, l_{i + 1}]]$, $\bigcup_{t \in Y_{i + 1}, l} U_t
  \subseteq \{p''_{i + 1, l} \}$. Thus we have, for each $l \in [[l_i + 1,
  l_{i + 1}]]$, $\bigcup_{t \in Y_{i + 1, l}} U_t =\{p''_{i + 1, l} \}$ and,
  hence, $| \bigcup_{l \in [[l_i + 1, l_{i + 1}]]} \bigcup_{t \in Y_{i + 1,
  l}} U_t | = l_{i + 1} - l_i$. This implies that, for each $l \in [[l_i, l_{i
  + 1} - 1]]$, $m''_{i + 1, l} > m''_{i + 1, l + 1} = p''_{i + 1, l + 1}$.
  
  We have $p''_i = (\tmop{bin} (x''_i) (1)^{l_{i \text{} + 1} - l_i}
  \gamma)_2 \leq (\tmop{bin} (x''_i + 1) (0)^{l_{i + 1} - l_i} \beta_{i +
  1})_2$ and $(\tmop{bin} (x''_i) (0)^{l_{i + 1} - l_i} \beta_{i + 1})_2 \leq
  r''$. Thus, $p''_{i + 1, l_i + 1} = (\tmop{bin} (x''_i) (1) (0)^{l_{i + 1} -
  l_i - 1} \beta_{i + 1})_2 $. Note also that, for $l \in [[l_i + 1, l_{i + 1}
  - 1]]$, we have that if $p''_{i + 1, l} = (\tmop{bin} (x''_i) (0)^{l - l_i -
  1} (1) (0)^{l_{i + 1} - l} \beta_{i + 1})_2$, then
  \begin{eqnarray*}
    \{p''_{i + 1, l + 1} \} & = & \bigcup_{t \in Y_{i + 1, l + 1}} U_t\\
    & \subseteq & X_{i + 1, l + 1} \cap [[r'' + 1, p''_{i + 1, l} - 1]]\\
    & \subseteq & X_{i + 1, l + 1} \cap [[(\tmop{bin} (x''_i) (0)^{l_{i + 1}
    - l_i} \beta_{i + 1})_2, p''_{i + 1, l} - 1]]\\
    & = & \{a_{l + 1} \},
  \end{eqnarray*}
  where $a_{l + 1} = (\tmop{bin} (x''_i) (0)^{l - l_i} (1) (0)^{l_{i + 1} - l
  - 1} \beta_{i + 1})_2$, and, hence, \ $p''_{i + 1, l + 1} = a_{l + 1}$.
  Thus, by induction, we have $x''_{i + 1, l_{i + 1}} = x''_{i + 1} =
  (\tmop{bin} (x''_i) (0)^{l_{i + 1} - l_i - 1} (1))_2$.
\end{proof}

\begin{lemma}
  \label{bin-lemma}If $t'' = s$ and, for some $i \in [[0, \tmop{last} - 1]]$,
  $l_{i + 1} = l_i + 1$ and $| \bigcup_{t \in Y_{i + 1}} U_t | = 2$, then $l_i
  > 0$, $x''_i > 0$ and $(\tmop{bin} (x''_{i + 1}))_2 = (\tmop{bin} (x''_i -
  1) (1))_2$ and $m''_{i + 1} = p''_{i + 1}$.
\end{lemma}

\begin{proof}
  If $l_i = 0$ and $l_{i + 1} = 1$, then, we have \ $i = 0$ (only $Y_0$ may
  have size $2^0$) and, since $t'' = s$, $m''_i \in X_i$, and, since $k \geq
  2$, for some $\alpha$, $X_i =\{((1) (0) \alpha)_2 \}$ and $X_{i + 1} =\{((0)
  (1) \alpha)_2, ((1) (1) \alpha)_2 \}$ and $((1) (1) \alpha)_2 > m''_i = ((1)
  (0) \alpha)_2$ and, hence, $| \bigcup_{t \in Y_{i + 1}} U_t | \leq 1$. Thus,
  we have $l_i > 0$.
  
  If $x''_i = 0$, then $m''_i \leq p''_i = \min X_i = ((0)^{l_i} \beta_i)_2$,
  where $\beta_i = (1) \gamma$ and $| \gamma | = | \beta_{i + 1} |$ and, by
  Lemma~\ref{suffix-lemma}, $(\gamma)_2 \leq (\beta_{i + 1})_2$. Thus $m''_i
  \leq ((0)^{l_i} (1) \beta_{i + 1})_2$ and, hence, $\bigcup_{t \in Y_{i + 1}}
  U_t \subseteq X_{i + 1} \cap [[r'' + 1, m''_i - 1]] \subseteq \{((0)^{l_i}
  (0) \beta_{i + 1})_2 \}=\{\min X_{i + 1} \}$. Thus we have $x''_i > 0$.
  
  We have $r'' \geq (\tmop{bin} (x''_i - 1) \beta_i)_2 \geq (\tmop{bin}
  (x''_i - 1) (1) (0)^{k - l_i - 1})_2 \geq a_i$, where $a_i = (\tmop{bin}
  (x''_i - 1) (0) \beta_{i + 1})_2$, since $(\beta_i)_2 \geq ((1) (0)^{k - l_i
  - 1})_2$.
  
  We have $p''_i = (\tmop{bin} (x''_i) \beta_i)_2 \leq b_i$, where $b_i =
  (\tmop{bin} (x''_i) (1) \beta_{i + 1})_2$, since $\beta_i = (1) \gamma$,
  where $| \gamma | = | \beta_{i + 1} |$ and, by Lemma~\ref{suffix-lemma},
  $\gamma \leq \beta_{i + 1}$ .
  
  If $r'' \geq (\tmop{bin} (x''_i - 1) (1) \beta_{i + 1})_2$, then $\bigcup_{t
  \in Y_{i + 1}} U_t \subseteq X_{i + 1} \cap [[r'' + 1, b_i - 1]] \subseteq$
  $\{(\tmop{bin} (x''_i) (0) \beta_{i + 1})_2 \}$. Thus, we have $a_i \leq r''
  < (\tmop{bin} (x''_i - 1) (1) \beta_{i + 1})_2$. In this case, \
  \begin{eqnarray*}
    \bigcup_{l \in [[0, l_i]]} X_{i + 1, l} \cap [[r'' + 1, m''_i - 1]] &
    \subseteq & \bigcup_{l \in [[0, l_i]]} X_{i + 1, l} \cap [[a_i + 1, b_i -
    1]]\\
    & = & \{(\tmop{bin} (x''_i) (0) \beta_{i + 1})_2 \}\\
    & = & \{p''_{i + 1, l_i} \}
  \end{eqnarray*}
  and $\bigcup_{l \in [[l_i + 1, l_{i + 1}]]} X_{i + 1, l} \cap [[r'' + 1,
  p''_{i + 1, l_i} - 1]] \subseteq \{(\tmop{bin} (x''_i - 1) (1) \beta_{i +
  1})_2 \}$. Hence, $| \bigcup_{t \in Y_{i + 1}} U_t | = 2$ implies \ that
  $p''_{i + 1} = p''_{i + 1, l_i + 1} = (\tmop{bin} (x''_i - 1) (1) \beta_{i +
  1})_2$ and $m''_{i + 1} = p''_{i + 1}$.
\end{proof}

\begin{lemma}
  \label{patch-lemma}If $t'' = s$ and, for some $i \in [[0, \tmop{last} -
  2]]$, $x''_{i + 1} = (\tmop{bin} (x''_i - 1) (1))_2$ and $m''_{i + 1} =
  p''_{i + 1}$ and, for some $d$ such that $d \geq 0$ and $i + 2 + d \leq
  \tmop{last}$, we have, for each $c \in [[0, d]]$, $| \bigcup_{t \in Y_{i + 2
  + c}} U_t | \geq l_{i + 2 + c} - l_{i + 1 + c}$, then, \ for each $c \in
  [[0, d]]$, we have $| \bigcup_{t \in Y_{i + 2 + c}} U_t | = l_{i + 2 + c} -
  l_{i + 1 + c} \leq 2$ and $x''_{i + 2 + c} = (\tmop{bin} (x''_{i + 1})
  \gamma_{i + 1} \ldots \gamma_{i + 1 + c})_2$, where $\gamma_j = (0)
  (1)^{l_{j + 1} - l_j - 1}$, and $m''_{i + 2 + c} = p''_{i + 2 + c}$.
\end{lemma}

\begin{proof}

  Note that $l_{i + 1} - l_i = 1$. We have $p''_i = (\tmop{bin} (x''_i)
  \beta_i)_2 = (\tmop{bin} (x''_i) (1) (0) \alpha_i)_2$. Hence, we have \ $r''
  \geq (\tmop{bin} (x''_i - 1) \beta_i)_2 = (\tmop{bin} (x''_i - 1) (1) (0)
  \alpha_i)_2 = (\tmop{bin} (x''_{i + 1}) (0) \alpha_i)_2$.
  
  Note that $(0) \alpha_i = \gamma_{i + 1} \ldots \gamma_{\tmop{last} - 1}$,
  where $\gamma_{i + 1 + c} = (0) (1)^{l_{i + 2 + c} - l_{i + 1 + c} - 1}$.
  
  The fact that $r'' \geq (\tmop{bin} (x''_{i + 1}) (0) \alpha_i)_2$ implies
  that, for arbitrary $c \in [[0, \tmop{last} - 1 - i]]$,
  \begin{equation}
    x''_{i + 1 + c} \geq (\tmop{bin} (x''_{i + 1}) \gamma_{i + 1} \ldots
    \gamma_{i + c})_2 \label{bottom}
  \end{equation}
  as follows:
  
  We have $\beta_{i + 1 + c} = (1)^{l_{i + 2 + c} - l_{i + 1 + c}} \gamma_{i
  + 2 + c} \ldots \gamma_{\tmop{last} - 1}$. Hence, for $x'' = (\tmop{bin}
  (x''_{i + 1}) \gamma_{i + 1} \ldots \gamma_{i + c})_2$, we have
  \begin{eqnarray*}
    2^{| \beta_{i + 1 + c} |} \cdot (x'' - 1) + (\beta_{i + 1 + c})_2 & < &
    2^{| \beta_{i + 1 + c} |} \cdot x'' + (\gamma_{i + 1 + c} \ldots
    \gamma_{\tmop{last} - 1})_2\\
    & = & (\tmop{bin} (x''_{i + 1}) \gamma_{i + 1} \ldots \gamma_{\tmop{last}
    - 1})_2\\
    & = & (\tmop{bin} (x''_{i + 1}) (0) \alpha_i)_2 .
  \end{eqnarray*}
  Thus $p''_{i + 1 + c} = \min \{x \in \mathbbm{X}_{i + 1 + c} | x > r'' \}
  \geq 2^{| \beta_{i + 1 + c} |} \cdot x'' + (\beta_{i + 1 + c})_2$ and,
  hence, $x''_{i + 1 + c} \geq x''$.
  
  To start induction, note that $x''_{i + 1} = (\tmop{bin} (x''_{i + 1})
  \gamma_{i + 1} \ldots \gamma_{i + 0})_2$, where $\gamma_{i + 1} \ldots
  \gamma_{i + 0}$ is an empty sequence.
  
  Induction step: We have $c \in [[0, d]]$, and
  \begin{itemize}
    \item  $x''_{i + 1 + c} = (\tmop{bin} (x''_{i + 1}) \gamma_{i + 1} \ldots
    \gamma_{i + c})_2$, and
    
    \item $| \bigcup_{t \in Y_{i + 1 + c}} U_t | \geq l_{i + 2 + c} - l_{i + 1
    + c}$.
  \end{itemize}
  \ Consider the case: $l_{i + 2 + c} - l_{i + 1 + c} = 1$. Then $\gamma_{i +
  1 + c} = (0)$. We also have
  \[ (\tmop{bin} (x''_{i + 1}) \gamma_{i + 1} \ldots \gamma_{i + c} \beta_{i
     + 1 + c})_2 = p''_{i + 1 + c} = m''_{i + 1 + c}, \]
  where $\beta_{i + 1 + c} = (1) \gamma_{i + 2 + c} \ldots \gamma_{\tmop{last}
  - 1}$. Note that $((1) \beta_{i + 2 + c})_2 \geq \beta_{i + 1 + c}$. Since
  $| \bigcup_{t \in Y_{i + 1 + c}} U_t | \geq 1$, we must have $m''_{i + 1 +
  c} > m''_{i + 2 + c} \in X_{i + 2 + c}$. Since $(\tmop{bin} (x''_{i + 1})
  \gamma_{i + 1} \ldots \gamma_{i + c} (1) \beta_{i + 2 + c})_2 \geq m''_{i +
  1 + c}$, and, by Equation~(\ref{bottom}),
  \[ x''_{i + 2 + c} \geq (\tmop{bin} (x''_{i + 1}) \gamma_{i + 1} \ldots
     \gamma_{i + 1 + c}) = (\tmop{bin} (x''_{i + 1}) \gamma_{i + 1} \ldots
     \gamma_{i + c} (0))_2, \]
  we must have
  \[ m''_{i + 2 + c} = p''_{i + 2 + c} = (\tmop{bin} (x''_{i + 1}) \gamma_{i
     + 1} \ldots \gamma_{i + c} (0) \beta_{i + 2 + c})_2 . \]
  Hence $x''_{i + 2 + c} = (\tmop{bin} (x''_{i + 1}) \gamma_{i + 1} \ldots
  \gamma_{i + 1 + c})_2$. \ By Lemma~\ref{bin-lemma}, we also have $|
  \bigcup_{t \in Y_{i + 1 + c}} U_t | = 1$.
  
  Consider the case: $l_{i + 2 + c} - l_{i + 1 + c} = 2$. Then $\gamma_{i + 1
  + c} = (01)$. Since, by Equation~(\ref{bottom}), $x''_{i + 2 + c} \geq
  (\tmop{bin} (x''_{i + 1 + c}) (01))_2$, and $| \bigcup_{t \in Y_{i + 2 + c}}
  U_t | \geq l_{i + 2 + c} - l_{i + 1 + c}$, we have, by
  Lemma~\ref{steps-lemma},
  \begin{eqnarray*}
    x''_{i + 2 + c} & = & (\tmop{bin} (x''_{i + 1}) \gamma_{i + 1} \ldots
    \gamma_{i + c} (01))_2\\
    & = & (\tmop{bin} (x''_{i + 1}) \gamma_{i + 1} \ldots \gamma_{i + 1 +
    c})_2
  \end{eqnarray*}
  and $p''_{i + 2 + c} = m''_{i + 2 + c}$ and $| \bigcup_{t \in Y_{i + 2 + c}}
  U_t | = l_{i + 2 + c} - l_{i + 1 + c}$.
  
  Consider the case: $l_{i + 2 + c} - l_{i + 1 + c} > 2$. Then $\gamma_{i + 1
  + c} = (0) (1)^q$, where $q = l_{i + 2 + c} - l_{i + 1 + c} - 1 > 1$. Since,
  by Equation~(\ref{bottom}), $x''_{i + 2 + c} \geq (\tmop{bin} (x''_{i + 1 +
  c}) (0)^q (1))_2$ and $| \bigcup_{t \in Y_{i + 1 + c}} U_t | \geq l_{i + 2 +
  c} - l_{i + 1 + c}$, we have by Lemma~\ref{steps-lemma}, $x''_{i + 2 + c} =
  (\tmop{bin} (x''_i) \gamma_{i + 1} \ldots \gamma_{i + c} (0)^q (1))_2$.
  However, by Equation~(\ref{bottom}), we have also
  \begin{eqnarray*}
    x''_{i + 2 + c} & \geq & (\tmop{bin} (x''_{i + 1}) \gamma_{i + 1} \ldots
    \gamma_{i + 1 + c})_2\\
    & = & (\tmop{bin} (x''_{i + 1}) \gamma_{i + 1} \ldots \gamma_{i + c} (0)
    (1)^q)_2\\
    & > & (\tmop{bin} (x''_{i + 1}) \gamma_{i + 1} \ldots \gamma_{i + c}
    (0)^q (1))_2,
  \end{eqnarray*}
  which is a contradiction. Thus the case \ $l_{i + 2 + c} - l_{i + 1 + c} >
  2$ is not possible.
\end{proof}

{\color{black} \begin{lemma}
  \label{right-side-lemma} $| \bigcup_{t \geq s} U_t | \leq k + 2$.
\end{lemma}

\begin{proof}
  If $t'' \not\in [[s, s + n - 1]]$ then $t'' = \infty$ and \ $r'' = n - 1$
  and $| \bigcup_{t \geq s} U_t | = 0$.
  
  Otherwise $t'' \in [[s, s + n - 1]]$. \ Note that w.l.o.g. we may assume
  that $t'' = s$, since $| \bigcup_{t \geq s} U_t | = | \bigcup_{t \geq t''}
  U_t |$. So, let us assume that $t'' = s$. We have $| \bigcup_{t \geq s} U_t
  | = \sum_{i = 0}^{\tmop{last}} | \bigcup_{t \in Y_{i + 1}} U_t | = |
  \bigcup_{t \in Y_0} U_t | + \sum_{i = 0}^{\tmop{last} - 1} | \bigcup_{t \in
  Y_{i + 1}} U_t |$. By Lemma~\ref{starting-lemma2}, $| \bigcup_{t \in Y_0}
  U_t | = 1 \leq l_0 + 1$.
  
  Let $V =\{i \in [[0, \tmop{last} - 1]] \; | \; | \bigcup_{t \in Y_{i + 1}}
  U_t | > l_{i + 1} - l_i \}$. By Lemma~\ref{last-lemma2}, $V \subseteq [[0,
  \tmop{last} - 2]]$. By Lemma~\ref{middle-lemma2}, if $i \in V$, then $l_{i +
  1} - l_i = 1$ and $| \bigcup_{t \in Y_{i + 1}} U_t | = 2$. By
  Lemma~\ref{bin-lemma}, if $i \in V$, then ($\tmop{bin} (x''_{i + 1}))_2 =
  (\tmop{bin} (x''_i - 1) (1))_2$.
  
  If $|V| \leq 1$, then $| \bigcup_{t \geq s} U_t | = | \bigcup_{t \in Y_0}
  U_t | + \sum_{i = 0}^{\tmop{last} - 1} | \bigcup_{t \in Y_{i + 1}} U_t |
  \leq (l_0 + 1) + \sum_{i = 0}^{\tmop{last} - 1} (l_{i + 1} - l_i) + 1 \leq k
  + 2$.
  
  Consider the case $|V| > 1$:
  
  Let $i \in V \setminus \{\max V\}$ and let $i' = \min \{j \in V \; | \; j >
  i\}$.
  
  Let $d = \max \{c \; | \; | \bigcup_{t \in Y_{i + 2 + c}} U_t | \geq l_{i +
  2 + c} - l_{i + 1 + c} \}$.
  
  Consider the case $i + 2 + d \geq i' + 1$: By Lemma~\ref{patch-lemma}, we
  have $| \bigcup_{t \in Y_{i' + 1}} U_t | = l_{i' + 1} - l_{i'}$. However,
  this is contradiction with $i' \in V$. Thus we have shown that there must be
  $i''$, such that $i < i'' < i'$ and $| \bigcup_{t \in Y_{i'' + 1}} U_t | <
  l_{i'' + 1} - l_{i''}$.
  
  Thus $\sum_{i = 0}^{\tmop{last} - 1} | \bigcup_{t \in Y_{i + 1}} U_t | \leq
  1 + \sum_{i = 0}^{\tmop{last} - 1} (l_{i + 1} - l_i)$ and, hence, $|
  \bigcup_{t \geq s} U_t | \leq k + 2$ also in the case $|V| > 1$.
\end{proof}}

\section{Bound on the extra-energy}

\begin{theorem}
  The extra-energy of RBO receiver is not greater than $2 k + 3$.
\end{theorem}

\begin{proof}
  The extra energy is equal to $| \bigcup_{t \geq s} U_t | + | \bigcup_{t \geq
  s} L_t |$, which by Lemmas~\ref{left-side-lemma} and \ref{right-side-lemma}
  is not greater than $2 k + 3$.
\end{proof}

\end{document}